\documentclass[onecolumn,12pt]{IEEEtran} %!PN
\usepackage{longtable}
\usepackage{multirow}
\usepackage{tabularx}
\usepackage{lipsum}
\usepackage{multicol}
\usepackage{amssymb}
\usepackage{amsmath}
\usepackage{longtable}
\newtheorem{theorem}{Theorem}[section]
\newtheorem{lemma}[theorem]{Lemma}
\newtheorem{proposition}[theorem]{Proposition}
\newtheorem{corollary}[theorem]{Corollary}
\newtheorem{example}{Example}
\newtheorem{remark}{Remark}

\usepackage{xcolor}

\newenvironment{proof}{\begin{IEEEproof}}{\end{IEEEproof}}

\makeatletter

\newcommand{\GF}[2][\cc]{{\mathbb F}_{{#1}^{#2}}}

\newcommand{\GFm}{\@ifstar{\GF m^\star}{\GF m}}
\newcommand{\Rmnum}[1]{\expandafter\@slowromancap\romannumeral #1@}

\makeatother

\ifCLASSINFOpdf
\else
\fi
\begin{document}
\title{Minimal linear codes from characteristic functions}

\author{Sihem Mesnager,
 \and Yanfeng Qi, \and Hongming Ru, \and Chunming Tang

\thanks{This work was supported by
the National Natural Science Foundation of China
(Grant No. 11871058, 11531002, 11701129).
S. Mesnager is  supported by the ANR CHIST-ERA project SECODE.
Y. Qi also acknowledges support from Zhejiang provincial Natural Science Foundation of China (LQ17A010008, LQ16A010005).
C. Tang also acknowledges support from
14E013, CXTD2014-4 and the Meritocracy Research Funds of China West Normal University.
}

\thanks{S. Mesnager is with the Department of Mathematics, University of Paris
VIII, 93526 Saint-Denis, France, with LAGA UMR 7539,  CNRS, Sorbonne Paris
Cit\'e, University of Paris XIII, 93430 Paris,  France, and also with Telecom
ParisTech, 75013 Paris, France (e-mail: smesnager@univ-paris8.fr).}

\thanks{Y. Qi is with School of Science, Hangzhou Dianzi University, Hangzhou,
Zhejiang, 310018, China (e-mail: qiyanfeng07@163.com).
}
\thanks{H. Ru and C. Tang are with the School of Mathematics  and Information,
China West Normal University, Nanchong 637002,  China (tangchunmingmath@163.com, hongming\_2436@163.com). C. Tang is also with the Department of Mathematics, The Hong Kong University
of Science and Technology, Clear Water Bay, Kowloon, Hong Kong.}

}

\maketitle

\begin{abstract}
Minimal linear codes have interesting applications in secret sharing schemes and secure two-party computation. This paper uses
characteristic functions of some subsets of $\mathbb{F}_q$
to construct minimal linear codes.
By properties of characteristic functions, we  can obtain more minimal binary linear codes from known minimal binary linear codes, which generalizes results of Ding et al. [IEEE Trans. Inf. Theory, vol. 64, no. 10, pp. 6536-6545,  2018].
By characteristic functions corresponding to some subspaces of $\mathbb{F}_q$, we obtain
many minimal linear codes, which generalizes
results of  [IEEE Trans. Inf. Theory, vol. 64, no. 10, pp. 6536-6545,  2018] and [IEEE Trans. Inf. Theory, vol. 65, no. 11, pp. 7067-7078,
    2019]. Finally, we use
characteristic functions to present a characterization of minimal linear codes from the defining set method and present  a class of minimal linear codes. 
\end{abstract}


% \fi

\IEEEpeerreviewmaketitle

\section{Introduction}
Throughout this paper, let
$p$ be a  prime and $q=p^m$, where $m$ is a positive integer. Let
$\mathbb{F}_q$ be the finite field with
$q$ elements and let
$\mathbb{F}_q^*$ be the multiplicative group of
$\mathbb{F}_q$.
 An $[n,k,d]$ linear code
$\mathcal{C}$ over $\mathbb{F}_q$ is a $k$-dimensional subspace of
$\mathbb{F}_q^n$ with minimum (Hamming) distance $d$.  Let $A_i$ be the number of codewords with Hamming weight $i$ in $\mathcal{C}$.
The weight enumerator of $\mathcal{C}$ is the polynomial
$1+A_1z+\cdots+A_nz^n$ and the weight distribution of $\mathcal{C}$ is $(1,A_1,\cdots,A_n)$.
The minimum distance $d$
determines the error-correcting capability of $\mathcal{C}$. The weight distribution contains
important information for estimating the probability of error detection and correction.  Hence,
 the weight distribution attracts much attention in coding theory and many papers focus on the determination of the weight distributions of linear codes.
 Let $t$ be the number of
nonzero $A_i$ in the weight distribution. Then
the code $\mathcal{C}$ is called a $t$-weight code.
Linear codes can be applied in consumer electronics, communication and data storage system.
Linear codes with  few weights   are   important in secret sharing \cite{CDY05,YD06}, authentication codes \cite{DHKW07,DW05}, association schemes \cite{CG84} and strongly regular graphs \cite{CK86}.

For a vector $a
=(a_1,\ldots,a_n)\in \mathbb{F}_q^n$,
let $Suppt(a)=\{1\leq i\leq n:
a_i\neq 0\}$ be the support of $a$ and let $wt(a)$ be the Hamming weight of $a$.
Note that $wt(a)=|Suppt(a)|$.
A vector $a\in \mathbb{F}_q^n$ covers
a vector $b\in \mathbb{F}_q^n$ if
$Suppt(b) \subseteq Suppt(a)$.
A codeword $a$ in a linear code $\mathcal{C}$ is minimal if $a$ covers only the codeword
$ua$ for all $u\in \mathbb{F}_q^*$, but no other
codewords in $\mathcal{C}$.
A linear code $\mathcal{C}$ is minimal if
any codeword of $\mathcal{C}$ is minimal.
Minimal linear codes have interesting applications in secret sharing schemes
\cite{CDY05,DY03,M93,YD06} and
secure two-party computation \cite{ABCH95,CCP14,CMP03}.
A  sufficient condition for a
linear code to be minimal is given in the following lemma.
\begin{lemma}\label{c-m}
\cite{AB98}
A linear code $\mathcal{C}$ over $\mathbb{F}_q$ is minimal if $\frac{w_{min}}{w_{max}}>\frac{q-1}{q}$, where
$w_{min}$ and $w_{max}$ denote the minimum and maximum nonzero Hamming weights in the code
$\mathcal{C}$ respectively.
\end{lemma}

Some minimal linear codes with few weights can be constructed by the defining set method
\cite{D15,D16}.  Let
$D=\{d_1,d_2,\ldots,d_n\}$ be a subset of  $\mathbb{F}_q$. Then a linear code of
length $n$  over $\mathbb{F}_p$ is defined by
\begin{equation}\label{def-c}
\mathcal{C}_D=
\{(Tr(\beta x))_{x\in D}: \beta
\in \mathbb{F}_q \},
\end{equation}
where $D$ is called the defining set of $\mathcal{C}_D$ and $Tr(x)=\sum_{i=0}^{m-1}x^{p^i}$ is the trace function from $\mathbb{F}_{q}$ to $\mathbb{F}_p$.
From this construction, many minimal linear codes can be constructed by different choices of $D$. Most of them satisfy  the sufficient condition $\frac{w_{min}}{w_{max}}>\frac{p-1}{p}$.
This sufficient condition is not necessary [7].
Chang and   Hyun \cite{CH18} made a breakthrough and constructed an infinite family of minimal binary linear codes with  $\frac{w_{min}}{w_{max}}<\frac{1}{2}$.
Heng et al. \cite{HDZ18} presented a sufficient and necessary
condition for minimal linear codes in the following theorem.
\begin{theorem}\label{thm-m}
Let $\mathcal{C}$ be a linear code over
$\mathbb{F}_p$. Then $\mathcal{C}$ is minimal if and only if
\begin{equation}\label{equ-m}
\sum_{c\in \mathbb{F}_p^*}wt(a+cb)
\neq (p-1)wt(a)-wt(b)
\end{equation}
for any $\mathbb{F}_p$ linearly independent
codewords $a,b\in \mathcal{C}$.
\end{theorem}
They also constructed an infinite family of minimal ternary linear codes with $\frac{w_{min}}{w_{max}}<\frac{2}{3}$.
Ding et al. \cite{DHZ18} presented more necessary and sufficient conditions for minimal binary linear codes and constructed three infinite families of minimal binary linear codes. Zhang et al. \cite{ZYW19}
constructed four families of minimal binary linear codes from Krawtchouk polynomials.
Xu and Qu \cite{XQ19} studied
minimal linear codes for odd $p$ and presented three infinite families of minimal linear codes. Bartoli and  Bonini \cite{BB19} generalized the third class of minimal linear codes in
\cite{DHZ18} from binary case to odd characteristic case and presented
a class of minimal linear codes
in odd characteristic.
Bonini and  Borello \cite{BB-19} presented many minimal linear codes from  particular blocking sets.
These minimal linear codes are constructed from the following method \cite{CD07,CH18,M17,MOS19,WHWK01}.
Let $f$ be a  function
from $\mathbb{F}_{q}$ to $\mathbb{F}_p$
such that
\begin{equation}\label{f-p}
\left\{
  \begin{array}{l}
    f(0)=0,   \\
    f(x)\neq Tr(wx)~\text{for all}~
w\in \mathbb{F}_q.
  \end{array}
\right.
\end{equation}
A linear code over $\mathbb{F}_p$
can  be defined by
\begin{equation}\label{c-f-1}
\mathcal{C}_f=\{(uf(x)-Tr(vx))_{x\in
\mathbb{F}_q^*}: u \in \mathbb{F}_p,
v\in \mathbb{F}_q\}
\end{equation}
By the choice of $f$, many linear codes with good properties can be defined.

Inspired by these recent results, we use the characteristic function of a subset of $\mathbb{F}_q$ to construct minimal linear codes in (\ref{c-f-1}). For binary case, by a simple property of characteristic functions, we can present more minimal binary linear codes from known minimal binary linear codes.
Furthermore, we employ characteristic functions corresponding to some subspaces to construct minimal linear codes, which generalize
\cite{DHZ18} and \cite{XQ19}.

The rest of this paper is organized as follows. In Section 2, we present some basic results on
$p$-ary functions, Krawchouk polynomials, and minimal linear codes. In Section 3, we present more minimal linear codes from characteristic functions. In Section 4, we use characteristic functions to present  a characterization of minimal linear codes from the defining set method and obtain a class of minimal linear codes. Section 5 makes a conclusion.

\section{Preliminaries}
In this section, we will introduce some results on $p$-ary functions, Krawchouk polynomials, and minimal linear codes.

\subsection{$p$-ary functions}
A $p$-ary function is a function from
$\mathbb{F}_{q}$ or $\mathbb{F}_p^m$ to $\mathbb{F}_{p}$.
The  Walsh transform of a $p$-ary function $f$ at a point $w\in\mathbb{F}_{q}$  is defined by
$$
\hat{f}(w):=
\sum_{x\in\mathbb{F}_{q}}\zeta_p^{f(x)
-Tr(w x)},
$$
where $\zeta_p=e^{2\pi \sqrt{-1}/p}$ is the primitive
$p$-th root of unity and $Tr$ is the trace function from $\mathbb{F}_q$ to
$\mathbb{F}_p$.
The  Walsh transform of a $p$-ary function $f$ at a point $w\in\mathbb{F}_{p}^m$  is defined by
$$
\hat{f}(w):=
\sum_{x\in\mathbb{F}_{p}^m}\zeta_p^{f(x)
-\langle w,x\rangle},
$$
where $\langle w, x \rangle$ is the inner product of $w$ and $x$.
A function $f(x)$ is called a $p$-ary bent function, if $|\hat{f}(w)|=p^{\frac{m}{2}}$ for any $w\in \mathbb{F}_q$.
When $p=2$, a $p$-ary (bent) function $f$ is just a
Boolean (bent) function.

An important class of Boolean functions is
the general Maiorana-McFarland class, which can be used to generate Boolean functions with good
cryptographic properties \cite{C10,CM16,D74,M73,MesnagerBook}.
Let $m$ be a positive integer and let
$s,t$ be two positive integers such that
$s+t=m$. The function in the general
Maiorana-McFarland class has the form
\begin{equation}\label{mm-f}
f(x,y)=\langle\phi(x), y\rangle+g(x),
\end{equation}
where $x\in \mathbb{F}_2^s$,
$y\in \mathbb{F}_2^t$, $\phi$ is a mapping
from $\mathbb{F}_2^s$ to $\mathbb{F}_2^t$,
and $g$ is a Boolean function in $s$ variables.

Krawchouk polynomials \cite{CZLH09,M77} are useful in bent functions and coding theory.
Let $m$ be a positive integer.
The  Krawchouk polynomial   is defined by
\begin{equation}\label{k-poly}
P_k(x)=\sum_{j=0}^k(-1)^j\binom{x}{j}
\binom{m-x}{k-j},
\end{equation}
where $0\leq k\leq m$.
The  Krawchouk polynomials satisfy
\begin{itemize}
\item $P_k(0)=\binom{m}{k}$,
\item $P_k(1)=\frac{m-2k}{m}\binom{m}{k}$,
\item $P_m(k)=(-1)^k$,
\item $P_k(i)=(-1)^iP_{m-k}(i)$ for $0\leq i\leq m$,
\item $\sum_{k=0}^{m}\binom{m-k}{m-j}P_k(x)
=2^j\binom{m-x}{j}$,
\item $P_k(x)=(-1)^kP_k(m-k)$.
\item $\sum_{wt(v)=k}(-1)^{u\cdot v}=P_k(i)$, where $u\in \mathbb{F}_2^m$ such that
$wt(u)=i$.
\end{itemize}

\subsection{Linear codes}
In this subsection, we present some results on  linear codes defined in (\ref{c-f-1}).

Parameters of binary linear codes in (\ref{c-f-1}) can be determined by the following Theorem.
\begin{theorem}[\cite{DHZ18}]\label{c-f-w}
Let $p=2$ and let
$\mathcal{C}_f$ be defined in
(\ref{c-f-1}) by a Boolean function
$f$ satisfying (\ref{f-p}).
The  code $\mathcal{C}_f$ has length
$q-1$ and dimension $m+1$.
The weight distribution of
$\mathcal{C}_f$ is given by the following
multiset union
$$
\left\{\frac{q-\hat{f}(w)}{2}:
w\in \mathbb{F}_q\right\}\cup
\left\{2^{m-1}: w\in \mathbb{F}_q^*\right\}
\cup \{0\}.
$$
\end{theorem}

A necessary and sufficient condition  of a minimal  binary linear code in
(\ref{c-f-1}) is given in the following theorem, which is more efficient than Theorem \ref{thm-m}.
\begin{theorem}[\cite{DHZ18}]\label{c-f-m}
Let $p=2$ and let $\mathcal{C}_f$ be defined in
(\ref{c-f-1}) from a Boolean function
$f$ satisfying (\ref{f-p}).
Then $\mathcal{C}_f$ is minimal if and only
if $\hat{f}(h)+\hat{f}(l)\neq q$
and $\hat{f}(h)-\hat{f}(l)\neq q$
for every pair of distinct elements
$h$ and $l$ in
$\mathbb{F}_{q}$.
\end{theorem}

Let $Q(\zeta_p)$ be the $p$-th
cyclotomic field over the rational field $Q$.
Then the field extension $Q(\zeta_p)/Q$ is
Galois of degree $p-1$ and the
Galois group is
$Gal(Q(\zeta_p)/Q)=
\{\sigma_a:  a\in \mathbb{F}_p^*\}$, where
$\sigma_a$ is an automorphism of
$Q(\zeta_p)$ defined by
$\sigma_a(\zeta_p)=\zeta_p^a$.
Parameters of a linear code in (\ref{c-f-1})
for odd $p$ can be given in the following lemma.
\begin{lemma}[\cite{M17}]\label{odd-c-f-w}
Let $p$ be an odd prime and let $\mathcal{C}_f$ be defined in
(\ref{c-f-1}). Then
$\mathcal{C}_f$ is a $[p^m-1,m+1]$ code and the Hamming weight of a codeword
$(uf(x)-Tr(vx))_{x\in
\mathbb{F}_q^*}$ is given by
$$
\left\{
  \begin{array}{ll}
    0, & \hbox{if}~u=0, v=0; \\
    p^m-p^{m-1}, & \hbox{if}~u=0, v\neq 0; \\
    p^m-p^{m-1}-\frac{1}{p}
\sum_{a\in \mathbb{F}_p^*}
\sigma_a\left(\sigma_u\left(\hat{f}(u^{-1}v)
\right)\right), & \hbox{otherwise.}
  \end{array}
\right.
$$
\end{lemma}

\section{Minimal linear codes from characteristic  functions}
In this section, we will present some minimal linear codes from characteristic functions associated with different subsets of $\mathbb{F}_q$.

Let $D\subset \mathbb{F}_q^*$. The characteristic  function of $D$ is
$$
f_D(x)=
\left\{
  \begin{array}{ll}
    1, & \hbox{if}~ x\in D  \\
    0, & \hbox{otherwise.}
  \end{array}
\right.
$$
From the characteristic function $f_D$, a linear code
$\mathcal{C}_{f_D}$ can be constructed by
\begin{equation}\label{c-f}
\mathcal{C}_{f_D}=\{(uf_D(x)-Tr(vx))_{x\in
\mathbb{F}_q^*}: u \in \mathbb{F}_p,
v\in \mathbb{F}_q\}.
\end{equation}

We first give some properties of characteristic functions.
\begin{lemma}\label{f-D}
Let $D\subset \mathbb{F}_q^*$ and let
$\overline{D}=\mathbb{F}_q^*\backslash D$.
Then
$$
\hat{f}_D(w)+\hat{f}_{\overline{D}}(w)=
\left\{
  \begin{array}{ll}
    (q-1)\zeta_p+q+1, & \hbox{if} ~w=0, \\
    1-\zeta_p, & \hbox{otherwise.}
  \end{array}
\right.
$$
\end{lemma}
\begin{proof}
\begin{align*}
\hat{f}_D(w) =& \sum_{x\in \mathbb{F}_q}\zeta_p^{
f_D(x)-Tr(wx)} \\
=& \sum_{x\in D}\zeta_p^{f_D(x)-Tr(wx)}+\sum_{x\in \mathbb{F}_q\backslash D}\zeta_p^{-Tr(wx)} \\
=& \sum_{x\in D}(\zeta_p^{f_D(x)-Tr(wx)}-
\zeta_p^{-Tr(wx)})+\sum_{x\in \mathbb{F}_q}\zeta_p^{-Tr(wx)}\nonumber\\
=& (\zeta_p-1) \sum_{x\in D}\zeta_p^{-Tr(wx)} +\sum_{x\in \mathbb{F}_{q}}\zeta_p^{-Tr(wx)}
\end{align*}
and
\begin{align*}
\hat{f}_{\overline{D}}(w) = (\zeta_p-1) \sum_{x\in \overline{D}}\zeta_p^{-Tr(wx)} +\sum_{x\in \mathbb{F}_{q}}\zeta_p^{-Tr(wx)}.
\end{align*}
Then
\begin{align*}
\hat{f}_{D}(w)+
\hat{f}_{\overline{D}}(w)=&
(\zeta_p-1) (\sum_{x\in {D}}\zeta_p^{-Tr(wx)}+
\sum_{x\in \overline{D}}\zeta_p^{-Tr(wx)}
) +2\sum_{x\in \mathbb{F}_{q}}\zeta_p^{-Tr(wx)}
\\
=& (\zeta_p-1) (\sum_{x\in \mathbb{F}_{q}}\zeta_p^{-Tr(wx)}-1)+2\sum_{x\in \mathbb{F}_{q}}\zeta_p^{-Tr(wx)}\\
=& 1-\zeta_p+(\zeta_p+1)\sum_{x\in \mathbb{F}_{q}}\zeta_p^{-Tr(wx)}.
\end{align*}
Since $\sum_{x\in \mathbb{F}_{q}}\zeta_p^{-Tr(wx)}=0$ for
any $w\in \mathbb{F}_q^*$, this lemma follows.
\end{proof}
\begin{lemma}\label{D12}
Let $D_1\subset \mathbb{F}_q^*$ and $D_2\subset \mathbb{F}_q^*$ such that $D_1\cap D_2=\emptyset$. Let $D=D_1\cup D_2$.
Then
$$
\hat{f}_D(w)=
\left\{
  \begin{array}{ll}
    \hat{f}_{D_1}(w)+ \hat{f}_{D_2}(w)-q, & \hbox{if} ~w=0, \\
    \hat{f}_{D_1}(w)+ \hat{f}_{D_2}(w), & \hbox{otherwise.}
  \end{array}
\right.
$$
\end{lemma}
\begin{proof}
Note that
\begin{align*}
\hat{f}_{D_1}(w)
=& (\zeta_p-1) \sum_{x\in D_1}\zeta_p^{-Tr(wx)} +\sum_{x\in \mathbb{F}_{q}}\zeta_p^{-Tr(wx)},\\
\hat{f}_{D_2}(w)
=& (\zeta_p-1) \sum_{x\in D_2}\zeta_p^{-Tr(wx)} +\sum_{x\in \mathbb{F}_{q}}\zeta_p^{-Tr(wx)},\\
\hat{f}_D(w)
=& (\zeta_p-1) \sum_{x\in D}\zeta_p^{-Tr(wx)} +\sum_{x\in \mathbb{F}_{q}}\zeta_p^{-Tr(wx)}.
\end{align*}
By $\sum_{x\in D}\zeta_p^{-Tr(wx)}
=\sum_{x\in D_1}\zeta_p^{-Tr(wx)}+
\sum_{x\in D_2}\zeta_p^{-Tr(wx)}$, we have this lemma.
\end{proof}
Using these properties of characteristic functions, we can give more linear codes $\mathcal{C}_{f_{\overline{D}}}$ from $\mathcal{C}_{f_D}$.
When $p=2$, we have   $\hat{f}_D(w)+\hat{f}_{\overline{D}}(w)=2$
for any $w\in \mathbb{F}_q$. If
  $D\subset \mathbb{F}_2^m\backslash \{\mathbf{0}\}$,
$\overline{D}=\mathbb{F}_2^m\backslash (\{\mathbf{0}\}\cup D)$, we also have $\hat{f}_D(w)+\hat{f}_{\overline{D}}(w)=2$ for any $w\in \mathbb{F}_p^m$.
By Theorem \ref{c-f-w} and Lemma
\ref{f-D}, we have the following corollary.
\begin{corollary}\label{fD2-w}
Let $p=2$. Let $D\subset \mathbb{F}_q^*$ and
$\overline{D}=\mathbb{F}_q^*\backslash D$ such that their characteristic functions $f_D$ and $f_{\overline{D}}$ satisfy
(\ref{f-p}).  Then
the  code $\mathcal{C}_{f_{\overline{D}}}$ has length
$q-1$ and dimension $m+1$.
The weight distribution of
$\mathcal{C}_{f_{\overline{D}}}$ is given by the following
multiset union
\begin{align*}
&\left\{\frac{q-\hat{f}_{\overline{D}}(w)}{2}:
w\in \mathbb{F}_q\right\}\cup
\{2^{m-1}: w\in \mathbb{F}_q^*\}
\cup \{0\}\\
=&\left\{\frac{q-2+\hat{f}_{{D}}(w)}{2}:
w\in \mathbb{F}_q\right\}\cup
\{2^{m-1}: w\in \mathbb{F}_q^*\}
\cup \{0\}.
\end{align*}
\end{corollary}
\begin{remark}
Let $f$ be a bent or semi-bent function satisfying (\ref{f-p}).  Let $D=Suppt(f)$.
The Walsh transforms of $f$ are given in \cite{CPT05}. By  Theorem \ref{c-f-m}, the codes
$\mathcal{C}_{f_D}$  and
$\mathcal{C}_{f_{\overline{D}}}$  are  minimal. They satisfy that $\frac{w_{min}}{w_{max}}\geq 1/2$.
\end{remark}
By Lemma  \ref{odd-c-f-w} and Lemma
\ref{f-D}, we have the following corollary.
\begin{corollary}\label{odd-fD2-w}
Let $p$ be an odd prime.  Let $D\subset \mathbb{F}_q^*$ and
$\overline{D}=\mathbb{F}_q^*\backslash D$ such that their characteristic functions $f_D$ and $f_{\overline{D}}$ satisfy
(\ref{f-p}).   Then
$\mathcal{C}_{f_{\overline{D}}}$ is a $[p^m-1,m+1]$ code and the Hamming weight of a codeword
$(uf_{\overline{D}}(x)-Tr(vx))_{x\in
\mathbb{F}_q^*}$ is given by
\begin{align*}
&\left\{
  \begin{array}{ll}
    0, & \hbox{if}~u=0, v=0; \\
    p^m-p^{m-1}, & \hbox{if}~u=0, v\neq 0; \\
    p^m-p^{m-1}-\frac{1}{p}
\sum_{a\in \mathbb{F}_p^*}
\sigma_a\left(\sigma_u\left(
\hat{f}_{\overline{D}}(0)
\right)\right), & \hbox{if}~u\neq 0, v= 0;\\
p^m-p^{m-1}-\frac{1}{p}
\sum_{a\in \mathbb{F}_p^*}
\sigma_a\left(\sigma_u\left(\hat{f}_{\overline{D}}(u^{-1}v)
\right)\right), & \hbox{if}~u\neq 0, v\neq 0.\\
  \end{array}
\right.\\
=&
\left\{
  \begin{array}{ll}
    0, & \hbox{if}~u=0, v=0; \\
    p^m-p^{m-1}, & \hbox{if}~u=0, v\neq 0; \\
    p^{m-1}-1+\frac{1}{p}
\sum_{a\in \mathbb{F}_p^*}
\sigma_a\left(\sigma_u\left(
\hat{f}_{{D}}(0)
\right)\right), & \hbox{if}~u\neq 0, v= 0;\\
p^m-p^{m-1}-1+\frac{1}{p}
\sum_{a\in \mathbb{F}_p^*}
\sigma_a\left(\sigma_u\left(\hat{f}_{{D}}(u^{-1}v)
\right)\right), & \hbox{if}~u\neq 0, v\neq 0.\\
  \end{array}
\right.
\end{align*}
\end{corollary}
\begin{remark}
By Corollary \ref{fD2-w} and
Corollary \ref{odd-fD2-w}, we can obtain
$\mathcal{C}_{f_{\overline{D}}}$ from known $\mathcal{C}_{f_D}$.
\end{remark}
In the following, we will use concrete
subsets $D$ to construct more minimal linear codes $\mathcal{C}_{f_D}$ and $\mathcal{C}_{f_{\overline{D}}}$.

\subsection{Some minimal binary linear codes  from known minimal binary linear codes}
In this subsection, we  will present more minimal binary linear codes from
known minimal binary linear codes in \cite{DHZ18}.

The following theorem generalizes
Theorem 23 in \cite{DHZ18} and obtains minimal linear codes from  Boolean functions
in the general Maiorana-McFarland class.
\begin{theorem}\label{c-mm}
Let $m\geq 7$ be an odd integer,
$s=\frac{m+1}{2}$, and $t=\frac{m-1}{2}$.
Let $U=\{x\in \mathbb{F}_2^s: wt(x)\geq 2\}$
and let $V=\{0\}$.
Let $f$ be the Boolean function defined in
(\ref{mm-f}), where $g\equiv 1$,
$\phi$ is an injection from
$\mathbb{F}_2^s\backslash U$ to
$\mathbb{F}_2^t\backslash V$, and
$\phi(x)=\mathbf{0}$ for any $x\in U$.
Let $D=Suppt(f)$.
The code $\mathcal{C}_{f_{\overline{D}}}$ defined in (\ref{c-f})
is a $[2^m-1,m+1, 2^{m-1}-2^{t-1}(2^s-s-1)-1]$ minimal code with  $\frac{w_{min}}{w_{max}}\leq 1/2$. The weight distribution of $\mathcal{C}_{f_{\overline{D}}}$ is given in
Table \ref{mm-w-o} (resp. Table
\ref{mm-w-e}) when
$s$ is odd (resp. even).
\begin{table}[htbp]
\caption{The weight distribution of $\mathcal{C}_{f_{\overline{D}}}$ in Theorem \ref{c-mm} for $s$ odd}\label{mm-w-o}
\centering
\begin{tabular}{|c|c|}
\hline
Weight & Frequency\\
\hline
0& 1\\
\hline
$2^{m-1}-1$ &
$2^s(2^t-s-2)+\binom{s}{(1+s)/2}$\\
\hline
$2^{m-1}$ & $2^m-1$\\
\hline
$2^{m-1}-2^{t-1}-1$& $s2^{s-1}$
\\
\hline
$2^{m-1}+2^{t-1}-1$& $2^s+s2^{s-1}$\\
\hline
$2^{m-1}+2^{t-1}(s+1-2i)-1$ for
$1\leq i\leq s$ and $i\neq (1+s)/2$&
$\binom{s}{i}$\\
\hline
$2^{m-1}-2^{t-1}(2^s-s-1)-1$& $1$\\
\hline
\end{tabular}
\end{table}
\begin{table}[htbp]
\caption{The weight distribution of $\mathcal{C}_{f_{\overline{D}}}$ in Theorem \ref{c-mm} for $s$ even}\label{mm-w-e}
\centering
\begin{tabular}{|c|c|}
\hline
Weight & Frequency\\
\hline
0& 1\\
\hline
$2^{m-1}-1$ &
$2^s(2^t-s-2)$\\
\hline
$2^{m-1}$ & $2^m-1$\\
\hline
$2^{m-1}-2^{t-1}-1$& $s2^{s-1}+\binom{s}{(s+2)/2}$
\\
\hline
$2^{m-1}+2^{t-1}-1$& $2^s+s2^{s-1}+\binom{s}{s/2}$\\
\hline
$2^{m-1}+2^{t-1}(s+1-2i)-1$ for
$1\leq i\leq s$ and $i\not \in \{s/2, (2+s)/2\}$&
$\binom{s}{i}$\\
\hline
$2^{m-1}-2^{t-1}(2^s-s-1)-1$& $1$\\
\hline
\end{tabular}
\end{table}
\end{theorem}
\begin{proof}
Note that
$$
\hat{f}_{D}(h_1,h_2)=
\left\{\begin{array}{l}{-2^{t}\left(2^{s}-s-1\right), \text { if } h_{1}=\mathbf{0} \text { and } h_{2}=\mathbf{0}} \\ {2^{t}(s+1-2 i), \text { if } h_{1} \neq \mathbf{0}, wt\left(h_{1}\right)=i \text { and } h_{2}=\mathbf{0}} \\ {-2^{t}(-1)^{h_{1} \cdot \phi^{-1}\left(h_{2}\right)}, \text { if } h_{2} \in \operatorname{Im} \phi \backslash\{\mathbf{0}\}} \\ {0, \text { if } h_{2} \notin \operatorname{Im} \phi}\end{array}\right.
$$
where $r$ runs from $1$ to $s$, and
$|\mathrm{Im}\phi|=s+2$.
By Lemma \ref{f-D},
we have
$$
\hat{f}_{\overline{D}}(h_1,h_2)=2-\hat{f}_{D}(h_1,h_2)=
\left\{\begin{array}{l}{2+2^{t}\left(2^{s}-s-1\right), \text { if } h_{1}=\mathbf{0} \text { and } h_{2}=\mathbf{0}} \\
2-{2^{t}(s+1-2 i), \text { if } h_{1} \neq \mathbf{0}, wt\left(h_{1}\right)=i \text { and } h_{2}=\mathbf{0}} \\
2+2^{t}(-1)^{h_{1} \cdot \phi^{-1}\left(h_{2}\right)}, \text { if } h_{2} \in \operatorname{Im} \phi \backslash\{\mathbf{0}\} \\
{2, \text { if } h_{2} \notin \operatorname{Im} \phi}\end{array}\right.
$$
Note that $\hat{f}_{\overline{D}}(h_1,h_2)
\pm\hat{f}_{\overline{D}}(l_1,l_2)\neq 2^m$ for
any pair of distinct
$(h_1,h_2)$, $(l_1,l_2)$.
By Theorem \ref{c-f-m},
The code $\mathcal{C}_{f_{\overline{D}}}$ is minimal.
By  Theorem \ref{c-f-w},
we have the weight distribution of $\mathcal{C}_{f_{\overline{D}}}$.
Note that $w_{min}=2^{m-1}-
2^{t-1}(2^s-s-1)-1$ and
$w_{max}=2^{m-1}+
2^{t-1}(s-1)-1$. Then $\frac{w_{min}}{w_{max}}\leq 1/2$.
This theorem follows.
\end{proof}

The following theorem generalizes Theorem 26 \cite{DHZ18}.
\begin{theorem}\label{c-f-D-k}
Let $k$ be a positive integer and let  $D=\{\alpha\in \mathbb{F}_2^m: 1\leq wt(\alpha)\leq k\}$.
The code $\mathcal{C}_{f_{\overline{D}}}$ defined in (\ref{c-f})
has length $2^m-1$, dimension $m+1$, and
the weight distribution in Table \ref{tfd1-k}.
\begin{table}[htbp]
\caption{The weight distribution of $\mathcal{C}_{f_{\overline{D}}}$ in Theorem \ref{c-f-D-k}}\label{tfd1-k}
\centering
\begin{tabular}{|c|c|}
\hline
Weight & Frequency\\
\hline
0& 1\\
\hline
$2^{m-1}-\sum_{j=1}^kP_j(i)-1$ for
$1\leq i\leq m$ &
$\binom{m}{i}$\\
\hline
$2^{m-1}$& $2^m-1$\\
\hline
$2^{m}-\sum_{j=1}^k\binom{m}{j}-1$&
$1$\\
\hline
\end{tabular}
\end{table}
\end{theorem}
\begin{proof}
Note that
$$
\hat{f}_{\overline{D}}(w)=
2-\hat{f}_{D}=
\left\{
  \begin{array}{ll}
    2-2^m+2\sum_{j=1}^k\binom{m}{j}, & \hbox{if} ~w=0 \\
    2+2\sum_{j=1}^kP_j(i), & \hbox{if}~w\neq 0~\text{and}~wt(w)=i
  \end{array}
\right.
$$
where $P_j(i)$ are  Krawchouk polynomials defined in (\ref{k-poly}).
By Theorem \ref{c-f-w}, we have the distribution of $\mathcal{C}_{f_{\overline{D}}}$ in Table  \ref{tfd1-k}.
\end{proof}
\begin{remark}
By Theorem \ref{c-f-m}, conditions of
$\mathcal{C}_{f_{\overline{D}}}$ to be minimal can be obtained.
\end{remark}

\subsection{Minimal linear codes from
characteristic functions corresponding to subspaces}
In this subsection, we will give some minimal linear codes from  characteristic functions
corresponding to some subspaces.

We first consider some subspaces in the following proposition.
\begin{proposition}\label{prop}
Let $E_1,\ldots, E_s$ be $s$ subspaces of
$\mathbb{F}_q$ such that
\begin{equation}\label{space-c}
\left\{
  \begin{array}{l}
    dim(E_i)=t_i, ~\forall ~1\leq i\leq s  \\
    E_i\cap E_j=\{0\},~ \forall~ 1\leq i\neq j\leq s\\
E_i^\bot\cap E_j^\bot=\{0\},~ \forall~ 1\leq i\neq j\leq s.
  \end{array}
\right.
\end{equation}
where $1\leq t_1\leq t_2\leq \cdots \leq t_s\leq m-1$. Then  one of the following conditions holds:

(i) $s=1$;

(ii) $s=2$ and $t_1+t_2=m$;

(iii) $s>2$, $m$  is even and  $t_1=\cdots=t_s=\frac{m}{2}$.
\end{proposition}
\begin{proof}
Conditions (i) and (ii) can be obtained when
$s=1$ or $s=2$.

Suppose that
$s>2$. If $t_1=dim(E_1)<\frac{m}{2}$,
by $dim(E_i)\cap dim(E_j)=\{0\}
~(i\neq j)$, then we have
$t_2>\frac{m}{2}, \ldots, t_s>\frac{m}{2}$, which makes a contradiction with $dim(E_2)\cap dim(E_s)=\{0\}$. Hence,
$t_1\geq \frac{m}{2}$. Similarly,
$t_2\geq \frac{m}{2}, \ldots, t_s\geq \frac{m}{2}$. By $dim(E_i)\cap dim(E_j)=\{0\}
~(i\neq j)$, we have
$t_1=t_2=\cdots=t_s=\frac{m}{2}$ and $m$ is even.

Hence, this proposition follows.
\end{proof}

Let $D=\cup_{i=1}^sE_i\backslash \{0\}$, where
$E_1,\ldots,E_s$ are subspaces satisfying (\ref{space-c}).
Note that
\begin{align}\label{fD-val}
\hat{f}_D(w) =& \sum_{x\in \mathbb{F}_q}\zeta_p^{
f_D(x)-Tr(wx)}\nonumber\\
=& \sum_{x\in D}\zeta_p^{f_D(x)-Tr(wx)}+\sum_{x\in \mathbb{F}_q\backslash D}\zeta_p^{-Tr(wx)}\nonumber\\
=& \sum_{x\in D}(\zeta_p^{f_D(x)-Tr(wx)}-
\zeta_p^{-Tr(wx)})+\sum_{x\in \mathbb{F}_q}\zeta_p^{-Tr(wx)}\nonumber\\
=& \sum_{i=1}^s(\zeta_p-1)(\sum_{x\in E_i}\zeta_p^{-Tr(wx)}-1)+\sum_{x\in \mathbb{F}_{q}}\zeta_p^{-Tr(wx)}\nonumber\\
=&
\left\{
  \begin{array}{ll}
    p^m+(\zeta_p-1)(\sum_{i=1}^s|E_i|-s), & \hbox{if}~w=0;  \\
    (\zeta_p-1)(|E_i|-s), & \hbox{if}~w\in E_i^{\bot}\backslash\{0\}
~\text{for}~1\leq i\leq s;\\
  -(\zeta_p-1)s, & \hbox{if}~ w\in
\mathbb{F}_q\backslash (\cup_{i=1}^sE_i^{\bot}),
  \end{array}
\right.
\end{align}
where $|E_i|=p^{t_i}$.
Then we have linear codes $\mathcal{C}_{f_D}$ and $\mathcal{C}_{f_{\overline{D}}}$
in the following theorem.
\begin{theorem}\label{thm-fD}
Let $D=\cup_{i=1}^sE_i\backslash \{0\}$, where
$E_1,\ldots, E_s$ satisfy (\ref{space-c}).
Let $\mathcal{C}_{f_D}$ and $\mathcal{C}_{f_{\overline{D}}}$ be defined in
(\ref{c-f}). Then
$\mathcal{C}_{f_D}$ and
$\mathcal{C}_{f_{\overline{D}}}$ are  $[q-1,m+1]$ codes with the weight distributions  in
Table \ref{W-1} and Table \ref{W-2}, respectively.
\begin{table}[htbp]
\caption{The weight distribution of $\mathcal{C}_{f_{{D}}}$ in Theorem \ref{thm-fD}}\label{W-1}
\centering
\begin{tabular}{|c|c|}
\hline
Weight & Frequency\\
\hline
0& 1\\
\hline
$p^{m}-p^{m-1}$ &
$p^m-1$\\
\hline
$\sum_{i=1}^sp^{t_i}-s$& $p-1$
\\
\hline
$p^{m}-p^{m-1}+p^{t_i}-s$
for $1\leq i\leq s$ & $(p-1)(p^{m-t_i}-1)$\\
\hline
$p^{m}-p^{m-1}-s$
&
$(p-1)(p^m-\sum_{i=1}^sp^{m-t_i}+s-1)$\\
\hline
\end{tabular}
\end{table}
\begin{table}[htbp]
\caption{The weight distribution of $\mathcal{C}_{f_{\overline{D}}}$ in Theorem \ref{thm-fD}}\label{W-2}
\centering
\begin{tabular}{|c|c|}
\hline
Weight & Frequency\\
\hline
0& 1\\
\hline
$p^{m}-p^{m-1}$ &
$p^m-1$\\
\hline
$p^m-1-\sum_{i=1}^sp^{t_i}+s$& $p-1$
\\
\hline
$p^{m}-p^{m-1}-1-p^{t_i}+s$
for $1\leq i\leq s$ & $(p-1)(p^{m-t_i}-1)$\\
\hline
$p^{m}-p^{m-1}-1+s$
&
$(p-1)(p^m-\sum_{i=1}^sp^{m-t_i}+s-1)$\\
\hline
\end{tabular}
\end{table}
\end{theorem}
\begin{proof}
By (\ref{fD-val}), Theorem \ref{c-f-w} and Corollary \ref{fD2-w}, for $p=2$,
we have the weight distributions of  $\mathcal{C}_{f_D}$ and $\mathcal{C}_{f_{\overline{D}}}$.
By (\ref{fD-val}),  Lemma \ref{odd-c-f-w} and
Corollary \ref{odd-fD2-w}, for $p$ odd,
we have the weight distributions of  $\mathcal{C}_{f_D}$ and  $\mathcal{C}_{f_{\overline{D}}}$. Hence, this theorem follows.
\end{proof}
By choosing different
subspaces $E_i$ in Theorem \ref{thm-fD}, we can obtain many minimal linear codes, in which we can find minimal codes with $\frac{w_{min}}{w_{max}}
\leq \frac{p-1}{p}$. Note that the codes $\mathcal{C}_{f_{{D}}}$ and $\mathcal{C}_{f_{\overline{D}}}$ can not be minimal if $s=1$.
We just consider
Condition (ii) and Condition (iii) in  Proposition \ref{prop}.

We first discuss linear codes satisfying
Condition (iii). When $p=2$ and $m$ is even, we have the following theorem on minimal linear codes.
\begin{theorem}[Theorem 18, \cite{DHZ18}]\label{p=2-t}
Let $p=2$, $m$ be even, and $s\geq 2$. Let $D=\cup_{i=1}^sE_i\backslash \{0\}$, where
$E_1,\ldots, E_s$ satisfy (\ref{space-c}),
$t_1=\cdots=t_s=t=\frac{m}{2}$.
Then $\mathcal{C}_{f_D}$ and
$\mathcal{C}_{f_{\overline{D}}}$ are minimal if and only if
$s\not \in \{2^{t}, 2^{t}+1\}$.
Furthermore, if $s\leq 2^{t-2}$,  the code $\mathcal{C}_{f_D}$ satisfies that $\frac{w_{min}}{w_{max}}
\leq \frac{1}{2}$. If $s> 3\cdot 2^{t-2}$, then
$\mathcal{C}_{f_{\overline{D}}}$ satisfies that $\frac{w_{min}}{w_{max}}
\leq \frac{1}{2}$.
\end{theorem}

For odd $p$, the following theorem gives minimal linear codes.
\begin{theorem}\label{podd-t}
Let $p$ be odd and $m$ be even.
 Let $D=\cup_{i=1}^sE_i\backslash \{0\}$, where
$E_1,\ldots, E_s$ satisfy (\ref{space-c})
 and
$t_1=\cdots=t_s=t=\frac{m}{2}$.
If  $p-2<s< p^{t}-p^{t-1}$ (resp. $s>p^{t-1}+1$),
then $\mathcal{C}_{f_D}$ (resp. $\mathcal{C}_{f_{\overline{D}}}$) is minimal.
Furthermore, if $s\leq p^t-2p^{t-1}
+p^{t-2}$ (resp. $s>2p^{t-1}-p^{t-2}$),  the code $\mathcal{C}_{f_D}$ (resp. $\mathcal{C}_{f_{\overline{D}}}$) satisfies that $\frac{w_{min}}{w_{max}}
\leq \frac{p-1}{p}$.
\end{theorem}
\begin{proof}
By the weight distribution of $\mathcal{C}_{f_D}$
in Table \ref{W-1}, we have
weights of nonzero codewords of  $\mathcal{C}_{f_D}$:
$w_1=sp^t-s$, $w_2=p^m-p^{m-1}-s$,
$w_3=p^m-p^{m-1}$, and $w_4=p^{m}-p^{m-1}
+p^t-s$. Obviously, $w_1<w_2<w_3< w_4$.
Let $H_i=\{wt(a)=w_i: a\in \mathcal{C}_{f_D}\}$
for $1\leq i\leq 4$.
Take two $\mathbb{F}_p$ linearly independent codewords $a=(u_1f(x)+Tr(v_1x))_{x
\in \mathbb{F}_q^*}, b=(u_2f(x)+Tr(v_2x))_{x
\in \mathbb{F}_q^*}$,
where $u_1, u_2\in \mathbb{F}_p$ and
$v_1,v_2\in \mathbb{F}_q$.
Note that
\begin{align*}
a\in H_1& ~\text{if and only if}~u_1\neq 0, v_1=0;\\
a\in H_2& ~\text{if and only if}~u_1\neq 0,
v_1\in \mathbb{F}_q\backslash \cup_{i=1}^{s}E_i^{\bot};  \\
a\in H_3& ~\text{if and only if}~u_1=0, v_1\neq 0;\\
a\in H_4& ~\text{if and only if}~u_1\neq 0,
 v_1\in \cup_{i=1}^{s}E_i^{\bot}\backslash\{0\}.
\end{align*}
By Theorem \ref{thm-m}, we just need to verify (\ref{equ-m}) for different cases of $a,b$.

Case 1: $a,b \in H_i$, where $i=1,2,3,4$.
Note that any two codewords in $H_1$ are linearly dependent and codewords with $u=0$ forms a one-weight code. Hence, (\ref{equ-m}) holds for $a,b\in H_1$ or $a,b\in H_3$.
We just consider
$a,b\in H_2$ or $a,b\in H_4$.
When $a,b\in H_2$, then $u_1,u_2, v_1, v_2\neq 0$. There exists only one $c\in \mathbb{F}_p^*$
such that $a+cb\in H_3$, and there exists
at most one $c\in \mathbb{F}_p^*$
such that $a+cb\in H_1$. Hence,
\begin{align*}
\sum_{c\in \mathbb{F}_q^*}wt(a+cb)
\geq w_1+(p-3)w_2+w_3>
(p-2)w_2=(p-1)wt(a)-wt(b).
\end{align*}
Similarly, when $a,b\in H_4$, by $s>p-2$,
\begin{align*}
\sum_{c\in \mathbb{F}_q^*}wt(a+cb)
\geq w_1+w_3+(p-3)w_2>
(p-2)w_4=(p-1)wt(a)-wt(b).
\end{align*}
Hence, (\ref{equ-m}) holds for $a,b\in H_i$ for
$1\leq i\leq 4$.

Case 2: $b\in H_1$, $a\in H_2$ or  $H_4$.
Suppose that $a\in H_2$. Then
there exists only one $c\in \mathbb{F}_p$ such that
$a+cb\in H_3$.  For other $c\in \mathbb{F}_p$,
$a+cb\in H_2$.
$$
\sum_{c\in \mathbb{F}_q^*}wt(a+cb)
=(p-2)w_2+w_3> (p-1)w_2-w_1
=(p-1)wt(a)-wt(b).
$$
This also holds for $a\in H_4$.

Case 3: $b\in H_1$, $a\in H_3$.
Then $a+cb\in H_2$ or $H_4$, where $c\in \mathbb{F}_p^*$. We have
$$
\sum_{c\in \mathbb{F}_q^*}wt(a+cb)
\geq (p-1)w_2> (p-1)w_3-w_1
=(p-1)wt(a)-wt(b).
$$

Case 4: $b\in H_2$, $a\in H_3$.
There exists at most  one $c\in \mathbb{F}_p$ such that $a+cb\in H_1$. For other $c\in  \mathbb{F}_p$, $a+cb\in H_2$ or $H_4$. We have
$$
\sum_{c\in \mathbb{F}_q^*}wt(a+cb)
\geq (p-2)w_2+w_1> (p-1)w_3-w_2
=(p-1)wt(a)-wt(b).
$$

Case 5: $b\in H_2$, $a\in H_4$. There exists only  one $c\in \mathbb{F}_p$ such that $a+cb\in H_3$. For other $c\in  \mathbb{F}_p$, $a+cb\in H_2$ or $H_4$. We have
$$
\sum_{c\in \mathbb{F}_q^*}wt(a+cb)
\geq (p-2)w_2+w_3> (p-1)w_4-w_2
=(p-1)wt(a)-wt(b).
$$

Case 6: $b\in H_3$, $a\in H_4$.
There exists at most one
$c_0\in \mathbb{F}_p^*$ such that
$a+c_0b\in H_1$. If such $c_0$ exists, then
$v_1+c_0v_2=0$ and $v_2=-\frac{1}{c_0}v_1\in \cup_{i=1}^{s}E_i^{\bot}\backslash\{0\}$.  For
$c\in \mathbb{F}_p^*\backslash \{c_0\}$,
$a+cb\in H_4$.
$$
\sum_{c\in \mathbb{F}_q^*}wt(a+cb)
\geq (p-2)w_4+w_1> (p-1)w_4-w_3
=(p-1)wt(a)-wt(b).
$$
If such $c_0$  does not exist, then $a+cb\in H_2$ or $H_4$.
$$
\sum_{c\in \mathbb{F}_q^*}wt(a+cb)
\geq (p-1)w_2> (p-1)w_4-w_3
=(p-1)wt(a)-wt(b).
$$

Hence, (\ref{equ-m}) holds.
By Theorem \ref{thm-m}, the code
$\mathcal{C}_{f_D}$ is minimal.
Furthermore, if $s\leq p^t-2p^{t-1}
+p^{t-2}$,  $\frac{w_{min}}{w_{max}}
=\frac{w_1}{w_4}\leq \frac{p-1}{p}$.

By the weight distribution of $\mathcal{C}_{f_{\overline{D}}}$
in Table \ref{W-2}, we have
weights of nonzero codewords of  $\mathcal{C}_{f_D}$:
$w_1'=p^m-1-sp^t+s$, $w_2'=p^m-p^{m-1}-1
-p^t+s$,
$w_3'=p^m-p^{m-1}$, $w_4'=p^m-p^{m-1}+s-1$.
By $s>p^{t-1}+1$, $w_1'<w_2'<w_3'<w_4'$.
Results on $\mathcal{C}_{f_{\overline{D}}}$  can be similarly obtained.
\end{proof}
\begin{remark}
Let $m=2t$. A partial spread of
$\mathbb{F}_q$ is a set of pairwise
supplementary $t$-dimensional subspaces of
$\mathbb{F}_q$.
Let $U_0, U_1,\ldots, U_{p^t}$ be a partial spread of $\mathbb{F}_q$, where
$U_i$ $(0\leq i\leq p^t)$ are $t$-dimensional subspaces of
$\mathbb{F}_q$. Take $s(p-1)$ subspaces
$E_1,\ldots, E_{s(p-1)}$ from
$U_0, U_1,\ldots, U_{p^t}$.
Let $D=\cup_{i=1}^{s(p-1)}\backslash \{0\}$.
Then $\mathcal{C}_{f_D}$ has the same
parameters and weight distribution with the third family of minimal linear codes in
\cite{XQ19}.
\end{remark}

In the following theorem, we will consider
minimal linear codes satisfying Condition
(ii) in Proposition \ref{prop}.
\begin{theorem}\label{s=2}
Let $p$ be a prime and let $D=(E_1\cup E_2)\backslash \{0\}$, where  $E_1,E_2$ are two subspaces of
$\mathbb{F}_q$ satisfying (\ref{space-c}),
$t_1+t_2=m$, and
 $2\leq t_1<t_2\leq m-2$.  Then  the codes $\mathcal{C}_{f_{{D}}}$ and $\mathcal{C}_{f_{\overline{D}}}$
 defined in (\ref{c-f}) are $[p^m-1, m+1]$ minimal codes  such that $\frac{w_{min}}{w_{max}}
\leq \frac{p-1}{p}$.
\end{theorem}
\begin{proof}
When $p=2$, by (\ref{fD-val}),  Theorem \ref{c-f-m}, and Theorem \ref{thm-fD},  the codes $\mathcal{C}_{f_{{D}}}$ and $\mathcal{C}_{f_{\overline{D}}}$
 defined in (\ref{c-f}) are $[2^m-1, m+1]$ minimal codes  such that $\frac{w_{min}}{w_{max}}
\leq \frac{1}{2}$.

When $p$ is odd, by Theorem \ref{thm-fD} and a similar proof with Theorem \ref{podd-t},  the codes $\mathcal{C}_{f_{{D}}}$ and $\mathcal{C}_{f_{\overline{D}}}$
 defined in (\ref{c-f}) are $[p^m-1, m+1]$ minimal codes  such that $\frac{w_{min}}{w_{max}}
\leq \frac{p-1}{p}$.
\end{proof}
\begin{remark}
Note that $E_1$ and $E_2$ can be identified as two linear codes over $\mathbb{F}_p$. By
 $E_1+E_2=\mathbb{F}_q$ and
$E_1\cap E_2=\{0\}$, $(E_1,E_2)$ is a linear complementary pair (LCP) of codes over $\mathbb{F}_p$ \cite{C15,CG16}.
We can take  two subspaces $E_1$ and
$E_2$ of $\mathbb{F}_q$, where
$E_2=E_1^\bot$ and $E_1+E_2=\mathbb{F}_q$.
Then $E_1$ is a linear complementary dual (LCD) code. There are many LCD codes constructed in \cite{CMTQ18,CMTQR18,MTQ18,ZLTD19}. Those LCD codes can be used in Theorem \ref{s=2} to construct minimal linear codes.
\end{remark}
\begin{example}
Let $p=2$ and let $m=5$.
Let $w$ be a primitive element of $\mathbb{F}_q$ such that $w^5+w^2+1=0$.
Take $E_{1}=\{0,w,w^9,w^{21}\}$ and
$E_2=\{0,w^6,w^7,w^{14},w^{18},w^{24},w^{26},
w^{29}\}$. Then $E_2=E_1^\bot$ and
$E_1+E_2=\mathbb{F}_q$.
The code $\mathcal{C}_{f_D}$ is a
minimal binary $[31,6,10]$ code with the weight enumerator $1+z^{10}
+21z^{14}+31z^{16}+7z^{18}
+3z^{22}$. The code $\mathcal{C}_{f_{\overline{D}}}$ is a
minimal binary $[31,6,9]$ code with the weight enumerator $1+3z^{9}
+7z^{13}+31z^{16}+21z^{17}
+z^{21}$.
\end{example}
\begin{example}
Let $p=3$ and let $m=5$.
Let $w$ be a primitive element of $\mathbb{F}_q$ such that $w^5+2w^2+1=0$.
Take $E_{1}$ as a subspace of $\mathbb{F}_q$ generated by $w^4$ and $w^{33}$.
Let $E_2=E_1^\bot$. Then  $E_1+E_2=\mathbb{F}_q$.
The code $\mathcal{C}_{f_D}$ is a
minimal  $[242,6,34]$ code with the weight enumerator $1+2z^{34}
+416z^{160}+242z^{162}+52z^{169}
+16z^{187}$. The code $\mathcal{C}_{f_{\overline{D}}}$ is a
minimal   $[242,6,136]$ code with the weight enumerator $1+16z^{136}
+52z^{154}+242z^{162}+416z^{163}
+2z^{208}$.
\end{example}

For some subspaces which do not satisfy
(\ref{space-c}), we have the following theorem on minimal linear codes.
\begin{theorem}\label{s=3}
Let $p$ be a prime and let $D= \cup_{i=1}^3 E_i\backslash \{0\}$, where  $E_1,E_2,E_3$ are three subspaces of
$\mathbb{F}_q$,
$E_j\cap E_j=\{0\}$ for
$1\leq i\neq j\leq 3$,
$\cap_{i=1}^3E_i^{\bot}
=\{0\}$,  and
 $1\leq t_1=t_2<t_3\leq m-2$.  Let
$t_{ij}$ be the dimension of
$E_i^{\bot}\cap E_j^{\bot}$
for $1\leq i\neq j\leq 3$.
Then  the codes $\mathcal{C}_{f_{{D}}}$ and $\mathcal{C}_{f_{\overline{D}}}$
 defined in (\ref{c-f}) are $[p^m-1, m+1]$  codes, whose   weight distributions are in
Table \ref{3W-1} and Table \ref{3W-2}, respectively.
Furthermore,  the code $\mathcal{C}_{f_{{D}}}$ is minimal such that $\frac{w_{min}}{w_{max}}
\leq \frac{p-1}{p}$.
\begin{table}[htbp]
\caption{The weight distribution of $\mathcal{C}_{f_{{D}}}$ in Theorem \ref{s=3}}\label{3W-1}
\centering
\begin{tabular}{|c|c|}
\hline
Weight & Frequency\\
\hline
0& 1\\
\hline
$p^{m}-p^{m-1}$ &
$p^m-1$\\
\hline
$\sum_{i=1}^3p^{t_i}-3$& $p-1$
\\
\hline
$p^{m}-p^{m-1}+p^{t_i}-3$
for $1\leq i\leq 3$ & $(p-1)(p^{m-t_i}-
\sum_{j\neq i}p^{t_{ij}}
+1)$ \\
\hline
$p^{m}-p^{m-1}+p^{t_i}+p^{t_j}-3$
for $1\leq i<j\leq 3$ & $(p-1)(p^{t_{ij}}-1)$\\
\hline
$p^{m}-p^{m-1}-3$
&
$(p-1)(p^m-|\cup_{i=1}^3E_i^{\bot}|)$\\
\hline
\end{tabular}
\end{table}
\begin{table}[htbp]
\caption{The weight distribution of $\mathcal{C}_{f_{\overline{D}}}$ in Theorem \ref{s=3}}\label{3W-2}
\centering
\begin{tabular}{|c|c|}
\hline
Weight & Frequency\\
\hline
0& 1\\
\hline
$p^{m}-p^{m-1}$ &
$p^m-1$\\
\hline
$p^m+2-\sum_{i=1}^3p^{t_i}$& $p-1$
\\
\hline
$p^{m}-p^{m-1}+2-p^{t_i}-p^{t_j}$
for $1\leq i< j\leq 3$ & $(p-1)(p^{m-t_{ij}}-
1)$
\\
\hline
$p^{m}-p^{m-1}+2-p^{t_i}$
for $1\leq i\leq 3$ & $(p-1)(p^{m-t_i}-
\sum_{j\neq i}p^{t_{ij}}
+1)$\\
\hline
$p^{m}-p^{m-1}+2$
&
$(p-1)(p^m-|\cup_{i=1}^3E_i^{\bot}|)$\\
\hline
\end{tabular}
\end{table}
\end{theorem}
\begin{proof}
Note that
\begin{align*}
\hat{f}_D(w)
=&
\left\{
  \begin{array}{ll}
    p^m+(\zeta_p-1)(\sum_{i=1}^s|E_i|-s), & \hbox{if}~w=0;  \\
    (\zeta_p-1)(|E_i|-s), & \hbox{if}~w\in E_i^{\bot}\backslash\{0\}
    ~\text{and}~w\not\in E_j^{\bot}
~\text{for}~1\leq i\leq s;\\
(\zeta_p-1)(|E_i|+|E_j|-s), & \hbox{if}~w\in (E_i^{\bot}\cap E_j^{\bot})\backslash\{0\}
~\text{for}~1\leq i\neq j\leq s;\\
  -(\zeta_p-1)s, & \hbox{if}~ w\in
\mathbb{F}_q\backslash (\cup_{i=1}^sE_i^{\bot}),
  \end{array}
\right.
\end{align*}
where $|E_i|=p^{t_i}$.
By a similar proof, this theorem follows.
\end{proof}
\begin{example}
Let $p=2$ and let $m=5$.
Let $w$ be a primitive element of $\mathbb{F}_q$ such that $w^5+w^2+1=0$.
Take $E_{1}=\langle w^{3}\rangle$, $E_2
=\langle w^4\rangle$ and $E_3=\langle
w^{6}, w^{10}, w^{28}\rangle$. Then
$E_i \cap E_j =\{0\}$ for $1\leq i\neq j\leq j$, $t_1=t_2=1$, $t_3=3$,
$dim(E_1^\bot\cap E_2^\bot)=3$,
$dim(E_1^\bot\cap E_3^\bot)=1$,
$dim(E_2^\bot\cap E_3^\bot)=1$,
and $\cap_{i=1}^3 E_i^\bot=\{0\}$.
The code $\mathcal{C}_{f_D}$ is a
minimal binary $[31,6,9]$ code with the weight enumerator $1+z^{9}
+7z^{13}+14z^{15}+31z^{16}+7z^{17}
+z^{21}+2z^{23}$. The code $\mathcal{C}_{f_{\overline{D}}}$ is a binary $[31,6,8]$ code with the weight enumerator $1+2z^{8}
+z^{10}+7z^{14}+45z^{16}+7z^{18}
+z^{22}$.
\end{example}
\begin{example}
Let $p=3$ and let $m=5$.
Let $w$ be a primitive element of $\mathbb{F}_q$ such that $w^5+2w^2+1=0$.
Take $E_{1}=\langle w^{75}\rangle$, $E_2
=\langle w^{223}\rangle$ and $E_3=\langle
w^{5}, w^{56}, w^{142}\rangle$. Then
$E_i \cap E_j =\{0\}$ for $1\leq i\neq j\leq j$, $t_1=t_2=1$, $t_3=3$,
$dim(E_1^\bot\cap E_2^\bot)=3$,
$dim(E_1^\bot\cap E_3^\bot)=1$,
$dim(E_2^\bot\cap E_3^\bot)=1$,
and $\cap_{i=1}^3 E_i^\bot=\{0\}$.
The code $\mathcal{C}_{f_D}$ is a
minimal binary $[242,6,30]$ code with the weight enumerator $1+2z^{30}
+208z^{159}+450z^{162}+52z^{165}+8z^{186}
+8z^{189}$. The code $\mathcal{C}_{f_{\overline{D}}}$ is a
minimal binary $[242,6,134]$ code with the weight enumerator $1+8z^{134}
+8z^{137}+52z^{158}+208z^{161}+242z^{162}
+208z^{164}+2z^{212}$.
\end{example}

\section{Minimal linear codes from the defining set  method}
In this section, by a defining set
$D\subset \mathbb{F}_q^*$,
we use the  characteristic function $f_D$ to
give a characterization of a minimal linear code  $\mathcal{C}_{D}$ in (\ref{def-c}).
For any $\beta\in \mathbb{F}_q$, let
$\mathbf{c}_{\beta}=(Tr(\beta x))_{x\in D}$.
For $\beta\neq 0$, note that
\begin{align}\label{fWT}
wt(\mathbf{c}_\beta)=& |D|
-
\frac{1}{p} \sum_{y\in \mathbb{F}_p}
\sum_{x\in D} \zeta_p^{-yTr(\beta x)}\nonumber\\
=& |D|-\frac{1}{p}|D|-
\frac{1}{p}\sum_{y\in \mathbb{F}_p^*}
\sum_{x\in D}\zeta_p^{-Tr(y\beta x)}\nonumber\\
=& \frac{p-1}{p}|D|
-\frac{1}{p\zeta_p}\sum_{y\in \mathbb{F}_p^*}
\sum_{x\in D}\zeta_p^{1-Tr(y\beta x)}\nonumber\\
=& \frac{p-1}{p}|D|
-\frac{1}{p\zeta_p}\sum_{y\in \mathbb{F}_p^*}
\sum_{x\in D}\zeta_p^{f_D(x)-Tr(y\beta x)}\nonumber\\
=& \frac{p-1}{p}|D|
-\frac{1}{p\zeta_p}\sum_{y\in \mathbb{F}_p^*}
(
\sum_{x\in \mathbb{F}_q}\zeta_p^{f_D(x)-Tr(y\beta x)}
-\sum_{x\in \mathbb{F}_q\backslash D}\zeta_p^{f_D(x)-Tr(y\beta x)})\nonumber\\
=& \frac{p-1}{p}|D|
-\frac{1}{p\zeta_p}\sum_{y\in \mathbb{F}_p^*}
(\hat{f}_D(y\beta)
-\sum_{x\in \mathbb{F}_q\backslash D}\zeta_p^{-Tr(y\beta x)})\nonumber\\
=& \frac{p-1}{p}|D|
-\frac{1}{p\zeta_p}\sum_{y\in \mathbb{F}_p^*}
(\hat{f}_D(y\beta)
-\sum_{x\in \mathbb{F}_q}\zeta_p^{-Tr(y\beta x)}
+\sum_{x\in D}\zeta_p^{-Tr(y\beta x)})\nonumber\\
=& \frac{p-1}{p}|D|
-\frac{1}{p\zeta_p}\sum_{y\in \mathbb{F}_p^*}
(\hat{f}_D(y\beta)+\sum_{x\in D}\zeta_p^{-Tr(y\beta x)})\nonumber\\
=& \frac{p-1}{p}|D|
-\frac{1}{p\zeta_p}\sum_{y\in \mathbb{F}_p^*}
(\hat{f}_D(y\beta)+\frac{1}{\zeta_p-1}
\hat{f}_D(y\beta))\nonumber\\
=& \frac{p-1}{p}|D|
-\frac{1}{p(\zeta_p-1)}\sum_{y\in \mathbb{F}_p^*}
\hat{f}_D(y\beta).
\end{align}
For any two $\beta_1,\beta_2\in \mathbb{F}_q^*$ and $c\in \mathbb{F}_p^*$, we have
$\mathbf{c}_{\beta_1}+c\mathbf{c}_{\beta_2}=
\mathbf{c}_{\beta_1+c\beta_2}$. By Theorem
\ref{thm-m}, we have the following theorem.
\begin{theorem}\label{m-definingset}
Let $D\subset \mathbb{F}_q^*$ and let
$\mathcal{C}_D$ be a linear code of dimension $m$  over
$\mathbb{F}_p$ defined in (\ref{def-c}). Then $\mathcal{C}_D$ is a minimal   code if and only if
\begin{equation}\label{Con-D}
\sum_{y\in \mathbb{F}_p^*}
(\sum_{c\in \mathbb{F}_p^*}
\hat{f}_D(y{\beta_1+yc\beta_2})
+\hat{f}_D(y{\beta_2})
-(p-1)\hat{f}_D(y{\beta_1}))
\neq (\zeta_p-1)(p-1)|D|
\end{equation}
for any $\mathbb{F}_p$ linearly independent
  ${\beta_1},
{\beta_2}\in \mathbb{F}_q^*$.
\end{theorem}
\begin{remark}
Take $D$ as a subset of $\mathbb{F}_p^m \backslash \{\mathbf{0}\}$.
By the defining method, we can also define a linear code $\mathcal{C}_D$ from $D$, where
the trace function is replace by the inner product.
For any $\beta\in \mathbb{F}_p^m$,
the weight of a codeword $\mathbf{c}_\beta$ can also be determined by the
Walsh transform of the characteristic function of $D$ in (\ref{fWT}).
Hence, Theorem \ref{m-definingset} also holds
for a subset $D$ of $\mathbb{F}_p^m \backslash \{\mathbf{0}\}$.
\end{remark}

When $p=2$, we have a characterization of minimal linear codes $\mathcal{C}_{D}$ and $\mathcal{C}_{\overline{D}}$.
\begin{theorem}\label{thm2-D-S}
Let $p=2$. Let $D\subset \mathbb{F}_q^*$ and let
$\mathcal{C}_D$ be a linear code  of dimension $m$ over
$\mathbb{F}_p$ defined in (\ref{def-c}). Then $\mathcal{C}_D$ is a minimal   code  if and only if
$\hat{f}_D({\beta_1+\beta_2})
-\hat{f}_D({\beta_1})
-\hat{f}_D({\beta_2})
\neq 2|D|$ for any
$\beta_1\neq \beta_2\in \mathbb{F}_q^*$.
Furthermore, if
$|\hat{f}_D(\beta)|<\frac{2}{3}|D|$ for any $\beta\in \mathbb{F}_q^*$,
then $\mathcal{C}$ is minimal.
\end{theorem}
\begin{corollary}
Let $p=2$. Let $D\subset \mathbb{F}_q^*$ and let
$\mathcal{C}_{\overline{D}}$ be a linear code  of dimension $m$ over
$\mathbb{F}_p$ defined in (\ref{def-c}), where
$\overline{D}=\mathbb{F}_q^*\backslash D$.
The code  $\mathcal{C}_{\overline{D}}$ is a minimal   code  if and only if
$\hat{f}_D({\beta_1})
+\hat{f}_D({\beta_2})
-\hat{f}_D({\beta_1+\beta_2})
\neq
2(2^m-|D|)$ for any
$\beta_1\neq \beta_2\in \mathbb{F}_q^*$.
\end{corollary}

In the following, we will give minimal linear codes from   subsets of $\mathbb{F}_p^m \backslash \{\mathbf{0}\}$.
Define
\begin{equation}\label{D12}
D_{12}=\{\beta: \beta\in \mathbb{F}_p^m,
1\leq wt(\beta)\leq 2\}.
\end{equation}
Then
\begin{align*}
|D_{12}|=&(p-1)\binom{m}{1}
+(p-1)^2\binom{m}{2}\\
=&\frac{p-1}{2}m(pm-p-m+3).
\end{align*}
For a $\beta=(b_0,\ldots,b_{m-1}) \in \mathbb{F}_p^m \backslash \{\mathbf{0}\}$,
let $t=wt(\beta)$, $s=m-t$ and
$A=\sum_{x\in D_{12}}\zeta_p^{-\langle\beta, x\rangle}$. We have $\hat{f}_{D_{12}}(\beta)
= (\zeta_p-1)
\sum_{x\in D_{12}}\zeta_p^{-\langle\beta, x\rangle}=(\zeta_p-1)A$ and 
\begin{align*}
A=& \sum_{x\in D_{12}, wt(x)=1}\zeta_p^{-\langle\beta, x\rangle}
+\sum_{x\in D_{12}, wt(x)=2}\zeta_p^{-\langle\beta, x\rangle}\\
=& \sum_{y\in \mathbb{F}_p^*}
\sum_{i=0}^{m-1}\zeta_p^{-yb_i}
+\sum_{y_1,y_2\in \mathbb{F}_p^*}
\sum_{0\leq i<j\leq m-1}\zeta_p^{-(y_1b_i+y_2b_j)}\\
=&
\sum_{i=0}^{m-1}\sum_{y\in \mathbb{F}_p^*}\zeta_p^{-yb_i}
+
\sum_{0\leq i<j\leq m-1}\sum_{y_1\in \mathbb{F}_p^*}\zeta_p^{-y_1b_i}
\sum_{y_2\in \mathbb{F}_p^*}
\zeta_p^{-y_2b_j}.
\end{align*}
Note that  $\sum_{y\in \mathbb{F}_p^*}\zeta_p^{-yb_i}=p-1$
for $b_i=0$ and $\sum_{y\in \mathbb{F}_p^*}\zeta_p^{-yb_i}=-1$
for $b_i\neq 0$. We have
\begin{align*}
A=&s(p-1)+t(-1)+ \binom{s}{2}(p-1)(p-1)
+st(-1)(p-1)+\binom{t}{2}(-1)(-1)\\
=& \frac{p^2}{2}t^2
-\frac{1}{2}(2p(p-1)m+4p-p^2)t+
\frac{p-1}{2}m(pm-p-m+3).
\end{align*}
We have $
\hat{f}_{D_{12}}(\beta)
= (\zeta_p-1)A
$, $\hat{f}_{D_{12}}(y\beta)=
\hat{f}_{D_{12}}(\beta)$~$(y\in \mathbb{F}_p^*)$,
and
\begin{align*}
wt(\mathbf{c}_\beta)=&
 \frac{p-1}{p}|D_{12}|
-\frac{1}{p(\zeta_p-1)}\sum_{y\in \mathbb{F}_p^*}
\hat{f}_{D_{12}}(y\beta)\\
=& \frac{p-1}{p}|D_{12}|-
\frac{p-1}{p}A\\
=& \frac{p-1}{2}\left(
-pt^2
+(2(p-1)m+4-p)t\right).
\end{align*}
Hence,   the weight of the  codeword
$\mathbf{c}_\beta$ is determined by the weight  $wt(\beta)$. Given $p$ and $m$, we can determine
the minimum and maximum nonzero 
Hamming weights. 
We need the following lemma to prove that 
the code $\mathcal{C}_{D_{12}}$ is minimal.
\begin{lemma}\label{lemma-D12}
Let $D$ be a subset of $\mathbb{F}_p^{m}$ satisfying $yD=D$ for any $y\in \mathbb{F}_p^*$. Then for any $\beta\in \mathbb{F}_p^*$,
\begin{equation}\label{sum-D}
\sum_{x\in D}\zeta_p^{-\langle\beta, x\rangle}
= \frac{1}{p-1}(
-|D|+p
|\{x\in D: \langle \beta ,x\rangle=0\}|
).
\end{equation}
\end{lemma}
\begin{proof}
For any $y\in \mathbb{F}_p^*$, $yD=D$.
We have
\begin{align*}
\sum_{x\in D}\zeta_p^{-\langle\beta, x\rangle}
=&\frac{1}{p-1}\sum_{x\in D}
\sum_{y\in \mathbb{F}_p^*}
\zeta_p^{-\langle y\beta, x\rangle}\\
=&\frac{1}{p-1}(
\sum_{x\in D, \langle \beta ,x\rangle=0}
\sum_{y\in \mathbb{F}_p^*}
\zeta_p^{-y\langle\beta, x\rangle}
+\sum_{x\in D, \langle \beta ,x\rangle\neq 0}
\sum_{y\in \mathbb{F}_p^*}
\zeta_p^{-y\langle\beta, x\rangle}
)\\
=&\frac{1}{p-1}(
\sum_{x\in D, \langle \beta ,x\rangle=0}
(p-1)
+\sum_{x\in D, \langle \beta ,x\rangle\neq 0}
(-1)
)\\
=& \frac{1}{p-1}(
-|D|+p
|\{x\in D: \langle \beta ,x\rangle=0\}|
).
\end{align*}
\end{proof}
Using Theorem \ref{thm2-D-S}, we have the following theorem on minimal linear codes
$\mathcal{C}_{D_{12}}$.

\begin{theorem}
Let $D_{12}$ be defined in (\ref{D12}). Then the code $\mathcal{C}_{D_{12}}$ is a minimal linear code.
Further, when $m\geq 6$,
the code $\mathcal{C}_{D_{12}}$ satisfies
$\frac{w_{min}}{w_{max}}\leq
\frac{p-1}{p}$.
\end{theorem}
\begin{proof}
By Theorem \ref{m-definingset}, we
just need to prove that
$\mathcal{C}_{D_{12}}$ satisfies
(\ref{Con-D}).
For any $\mathbb{F}_p$ linearly independent
  ${\beta_1},
{\beta_2}\in \mathbb{F}_p^m$, define
$$
M=\sum_{y\in \mathbb{F}_p^*}
(\sum_{c\in \mathbb{F}_p^*}
\hat{f}_{D_{12}}(y{\beta_1+yc\beta_2})
+\hat{f}_{D_{12}}(y{\beta_2})
-(p-1)\hat{f}_{D_{12}}(y{\beta_1})).
$$
Note that $D_{12}$ satisfies Lemma \ref{lemma-D12},  $\hat{f}_{D_{12}}(\beta)
= (\zeta_p-1)
\sum_{x\in D_{12}}\zeta_p^{-\langle\beta, x\rangle}$, and $\hat{f}_{D_{12}}(y\beta)
=\hat{f}_{D_{12}}(\beta)$, where
$y\in \mathbb{F}_p^*$. We have
\begin{align*}
M=& (p-1)(\sum_{c\in \mathbb{F}_p^*}
\hat{f}_{D_{12}}({\beta_1+c\beta_2})
+\hat{f}_{D_{12}}({\beta_2})
-(p-1)\hat{f}_{D_{12}}({\beta_1}))\\
=&(p-1)(\zeta_p-1)
(A_1
+A_2
-A_3)
\end{align*}
where $A_1=\sum_{c\in \mathbb{F}_p^*}
\sum_{x\in D_{12}}
\zeta_p^{-\langle \beta_1+c\beta_2,x\rangle}$,
$A_2=
\sum_{x\in D_{12}}
\zeta_p^{-\langle \beta_2,x\rangle}$ and
$A_3=(p-1)
\sum_{x\in D_{12}}
\zeta_p^{-\langle \beta_1,x\rangle}$.
Note that
\begin{align*}
A_1=&
\sum_{x\in D_{12}}
\zeta_p^{-\langle \beta_1,x\rangle}\sum_{c\in \mathbb{F}_p^*}
\zeta_p^{-\langle c\beta_2,x\rangle}\\
=&
\sum_{x\in D_{12},
\langle \beta_2,x\rangle=0}
(p-1)\zeta_p^{-\langle \beta_1,x\rangle}
+\sum_{x\in D_{12},
\langle \beta_2,x\rangle\neq 0}
(-1)\zeta_p^{-\langle \beta_1,x\rangle} \\
=&
p\sum_{x\in D_{12},
\langle \beta_2,x\rangle=0}
\zeta_p^{-\langle \beta_1,x\rangle}
-\sum_{x\in D_{12}}
\zeta_p^{-\langle \beta_1,x\rangle}.
\end{align*}
By Lemma \ref{lemma-D12},  we have
\begin{align*}
M=& (p-1)(\zeta_p-1)\left(
p(\sum_{x\in D_{12},
\langle \beta_2,x\rangle=0}
\zeta_p^{-\langle \beta_1,x\rangle}
-\sum_{x\in D_{12}}
\zeta_p^{-\langle \beta_1,x\rangle})
+\sum_{x\in D_{12}}
\zeta_p^{-\langle \beta_2,x\rangle}
\right)\\
=& (p-1)(\zeta_p-1)\left(
\sum_{x\in D_{12}}
\zeta_p^{-\langle \beta_2,x\rangle}-
p\sum_{x\in D_{12},
\langle \beta_2,x\rangle\neq 0}
\zeta_p^{-\langle \beta_1,x\rangle}
\right)\\
=& (\zeta_p-1)\Big(
-|D_{12}|+p
|\{x\in D_{12}: \langle \beta_2 ,x\rangle=0\}|
\\
&
\quad\quad\quad\quad
-p(-|\{x\in D_{12}: \langle \beta_2 ,x\rangle\neq 0\}|
+p|\{x\in D_{12}: \langle \beta_2 ,x\rangle\neq 0,
\langle \beta_1 ,x\rangle=0\}|
)\Big) \\
=& (\zeta_p-1)\left(
(p-1)|D_{12}|-p^2|\{x\in D_{12}: \langle \beta_2 ,x\rangle\neq 0,
\langle \beta_1 ,x\rangle=0\}|
\right).
\end{align*}
Hence, we need to prove that
$|\{x\in D_{12}: \langle \beta_2 ,x\rangle\neq 0,
\langle \beta_1 ,x\rangle=0\}|>0$.
Since ${\beta_1}=(b_{11}, \ldots,
b_{1m}),
{\beta_2}=(b_{21}, \ldots,
b_{2m})\in \mathbb{F}_p^m$ are   linearly independent, there exist $i$ and $j$ such that
$(b_{1i}, b_{2i})$ and
$(b_{1j}, b_{2j})$ are linearly independent, where $1\leq i\neq j\leq m$. Then
there exist $c_1,c_2\in \mathbb{F}_p$ 
such that $c_1(b_{1i}, b_{2i}) +c_2(b_{1j}, b_{2j})=(0,1)$.
Take a vector $v=(v_1,\ldots,v_m)\in \mathbb{F}_p^m$, where
$v_i=c_1$, $v_j=c_2$, and
$v_k=0$ for $k\neq i,j$.
Then $wt(v)=1$ or $2$,  and
$v\in D_{12}$ such that
$\langle\beta_2, v\rangle \neq 0,
\langle \beta_1, v\rangle=0$.
Hence, $|\{x\in D_{12}: \langle \beta_2 ,x\rangle\neq 0,
\langle \beta_1 ,x\rangle=0\}|>0$ and
$M\neq (\zeta_p-1)(p-1)|D_{12}|$. By Theorem
\ref{m-definingset}, $\mathcal{C}_{D_{12}}$ is a minimal linear code.

The weight of the  codeword
$\mathbf{c}_\beta$ is determined by the weight  $wt(\beta)$.
When $wt(\beta)=1$,
$\mathbf{c}_\beta=(p-1)(pm-m-p+2)$.
When $wt(\beta)=m$,
$\mathbf{c}_\beta=\frac{1}{2}(p-1)m(pm-2m-p+4)$. When $p=2$ and $m\geq 6$, we can verify that $\frac{w_{min}}{w_{max}}\leq \frac{1}{2}$.
When $p\geq 3$ and $m\geq 5$, we have
\begin{align*}
\frac{w_{min}}{w_{max}}\leq&
\frac{2(pm-m-p+2)}{m(pm-2m-p+4)}
\\
\leq& \frac{2(p-1)}{(p-2)m-p+4} \\
\leq & \frac{2(p-1)}{(p-2)\cdot 5-p+4}\\
\leq & \frac{p-1}{p}.
\end{align*}
Hence, this theorem follows.
\end{proof}
\begin{remark}
When $p=2$, these codes have been
studied in \cite{ZYW19}.
\end{remark}

\begin{example}
Let $p=2$ and let $m=5$.
Take $D=\{x\in \mathbb{F}_p^m \backslash
\{\mathbf{0}\}: wt(x)=1,2,m\}$.
Then $|\hat{f}_D(\beta)|<\frac{2}{3}|D|$ for any $\beta\in \mathbb{F}_q^*$.
The code $\mathcal{C}_{f_D}$ is   a
minimal binary $[16,5,6]$ code with the weight enumerator $1+6z^{6}
+15z^{8}+10z^{10}$. The code $\mathcal{C}_{f_{\overline{D}}}$ is a
minimal binary $[15,5,6]$ code with the weight enumerator $1+10z^{6}
+15z^{8}+6z^{10}$.
\end{example}
\begin{example}
Let $p=3$ and let $m=6$.
Take $D_{12}=\{x\in \mathbb{F}_p^m \backslash
\{\mathbf{0}\}: wt(x)=1,2\}$
The code $\mathcal{C}_{f_{D_{12}}}$ is   a
minimal $[72,6,22]$ code with the weight enumerator $1+12z^{22}+60z^{38}
+64z^{42}+160z^{48}+192z^{50}
+240z^{52}$, and it  satisfies that $\frac{w_{min}}{w_{max}}
\leq \frac{2}{3}$.
\end{example}
\begin{example}
Let $p=5$ and let $m=4$.
Take $D_{12}=\{x\in \mathbb{F}_p^m \backslash
\{\mathbf{0}\}: wt(x)=1,2\}$
The code $\mathcal{C}_{f_{D_{12}}}$ is   a
minimal $[112,4,52]$ code with the weight enumerator $1+16z^{52}+96z^{84}
+ 256z^{88}
+256z^{96}$, and it  satisfies that $\frac{w_{min}}{w_{max}}
\leq \frac{4}{5}$.
\end{example}

\section{Conclusion}
By characteristic functions of subsets of
$\mathbb{F}_q$, we can construct more minimal linear codes, which generalizes
results in \cite{DHZ18} for the binary case and
\cite{XQ19} for $p$ odd.
These minimal linear codes satisfy
$\frac{w_{min}}{w_{max}}\leq \frac{p-1}{p}$.
It is interesting to construct more minimal linear codes. We also use characteristic functions to present  a characterization of minimal linear codes from the defining set method.
Theorem \ref{c-f-m} is efficient to determine a minimal binary linear code.
Theorem \ref{thm-m} is not efficient enough for a minimal linear code for $p$ odd.
It would be interesting to present more efficient
results
to determine minimal linear codes for $p$ odd.

%\begin{IEEEbiography}
%\end{IEEEbiography}


%\newpage


%\vfill


\end{document}